%% file: StringInstability.tex
\documentclass[technote, 10pt, twoside]{IEEEtran}   
\usepackage{cite} 
\pdfoutput=1
\ifCLASSINFOpdf 
  \usepackage[pdftex]{graphicx}
\else
  \usepackage[dvips]{graphicx}  
\fi
 
\usepackage[cmex10]{amsmath} 
\usepackage{amsfonts}
\usepackage{amsthm}
\usepackage{mathtools}  
\usepackage{algorithmic}  
\usepackage[hidelinks]{hyperref} 
\usepackage{array}
\usepackage{stfloats}
\usepackage{subcaption}
\usepackage{url}
\usepackage{tikz}
\usepackage{mathtools}
\usetikzlibrary{arrows, decorations.pathmorphing, backgrounds, positioning,
fit, shapes.geometric} 

\usepackage[finalold]{shared/trackchanges}

\addeditor{ZH} 
\addeditor{IH}

\include{shared/notation}

\include{shared/MyTikzStyles}

\newtheoremstyle{assump}
{1pt}
{1pt}
{}
{}
{}
{)}
{-1.8em}
{}

\theoremstyle{plain}
\newtheorem{theorem}{Theorem}
\newtheorem{lemma}{Lemma}
\newtheorem{corollary}{Corollary}
\theoremstyle{definition}
\newtheorem{definition}{Definition}

\theoremstyle{assump}

\begin{document}
%
\title{Nonzero bound on Fiedler eigenvalue causes exponential growth of
H-infinity norm of vehicular platoon}
%

        \author{Ivo~Herman, Dan~Martinec,
        Zden\v{e}k~Hur\'{a}k, Michael \v{S}ebek
\thanks{All authors are with the Faculty of Electrical Engineering, Czech
Technical University in Prague, Department of Control Engineering, Karlovo
namesti 13, Prague, Czech Republic.
E-mail address:
\{ivo.herman, martinec.dan, hurak, sebekm1\}@fel.cvut.cz} \thanks{The research
was supported by the Grant Agency of the Czech Republic within the projects GACR 13-06894S (I.~H.) and GACR P103-12-1794 (D.~M. and M.~S.). Involvement of Z.H. was supported by Fulbright Program.
}
\thanks{Manuscript received December 04, 2013; revised August 12, 2014.}

\thanks{\copyright 2014 IEEE. Personal use of this material is permitted.
Permission from IEEE must be obtained for all other uses, in any current or future media, including reprinting/republishing this material for advertising or promotional purposes, creating new collective works, for resale or redistribution to servers or lists, or reuse of any copyrighted component of this work in other works.}}


\maketitle

\begin{abstract} 
We consider platoons composed of identical vehicles and controlled in a
distributed way, that is, each vehicle has its own onboard controller. The
regulation errors in spacing to the immediately preceeding and following
vehicles are weighted differently by the onboard controller, which thus implements an
asymmetric bidirectional control scheme. The weights can vary along the platoon.
We prove that such platoons have a nonzero uniform bound on the second smallest
eigenvalue of the graph Laplacian matrix---the \textit{Fiedler eigenvalue}.
Furthermore, it is shown that existence of this bound always signals undesirable
scaling properties of the platoon. Namely, the $\mathcal{H}_\infty$ norm of the
transfer function of the platoon grows exponentially with the number of vehicles
regardless of the controllers used.
Hence the benefits of a uniform gap in the spectrum of a Laplacian with an
asymetric distributed controller 
are paid for by
poor scaling as the number of vehicles grows.
\end{abstract}


\IEEEpeerreviewmaketitle

\section{Introduction} 
\IEEEPARstart{P}{latoons} of vehicles offer promising solutions for future
highway transport. They provide several advantages to the current highway
traffic---they increase both the capacity and the safety of the highway and they
allow the driver to relax\note[ZH]{Tady by docela patrily dve nebo tri citace
nejakych prehledovych clanku. Vice viz poznamky v prilozenem
textu}\note[IH]{To sice jo, ale rozhodne na to neni prostor, uz tak musim
nekde odebrat ctvrt stranky. Jedna se vejde, ale nevim jaka}.

Several platoon control architectures have been proposed in the literature. They
differ mainly in presence of direct interactions with the platoon leader. If the
information from the leader is permanently available to all following vehicles,
the platoon can behave very well and is scalable. On the other hand, it requires
establishing some communication among the vehicles, which can be disturbed or
even denied by an intruder. Control schemes relying on communication comprise
\textit{leader following} and \textit{cooperative adaptive cruise control}. For
their overview and properties see e.~g., \cite{Seiler2004a,
Sebek2011, Milanes2014}\note[ZH]{Tady bych to taky jeste o jednu dve tri citace
aktualnich prehledovych praci rozsiril, vice v poznamkach.}\note[IH]{Pridana
citace Milanes2014, kde je seznam praktickych projektu i nejake teoreticke
prace}.

Among the communication-free scenarios are the \textit{predecessor following},
\textit{constant time-headway spacing} and \textit{bidirectional control}
(\textit{symmetric} or \textit{asymmetric}). Recognizing their limitations,
these architectures may still be useful as backup control solutions during
communication failures. One of the key theoretical issues investigated with
communication-free control schemes is \annote[ZH]{\textit{string stability}}{K
tomuto viz poznamky.}. Although there are variations among the concepts found in
the literature (for a review see, e.g., \cite{Ploeg2014}), the key idea is that
the platoon is \textit{string unstable} if the impact of a disturbance affecting one vehicle gets amplified as it
propagates along the string (platoon). The predecessor-following strategy is
string unstable if there are at least two integrators in the open loop of each
vehicle \cite{Seiler2004a}. \annote[ZH]{Two integrators are a reasonable
assumption}{K tomuto diskuze v samostatnych poznamkach}, as they allow both
velocity tracking and constant spacing \cite{Barooah2005}. The constant
time-headway spacing policy increases the required intervehicular distances in
response to the increased speed of the leader, which preserves the string
stability \cite{Middleton2010}. For symmetric bidirectional formations, the
response to noise (coherence) scales polynomially with the size of the platoon
\cite{Bamieh2012}. The paper also reveals the bad effect of increasing the number of integrators in the open loop.
 
Asymmetric controllers for platoons have received much attention after
\cite{Barooah2009a} was published. The authors show that for small controller
asymmetry,
the convergence rate of the least stable eigenvalue to zero
(as the number of vehicles grows) decreases. 
Later the paper \cite{Hao2012} shows that with a non-vanishingly small
asymmetry, the least stable eigenvalue does not actually converge to zero but to
some nonzero constant---a~uniform nonzero lower bound can be achieved.
\note[IH]{Tady jsem popis uniform boundedness vyhodil - Uniform boundedness
means that the bound on the Laplacian eigenvalues does not depend on the number of vehicles. }{}
This result guarantees a
controllability of the formation consisting of an arbitrary number of vehicles.
In \cite{Lin2012} optimal localized control for asymmetric formation is
proposed. The authors show that asymmetric control has beneficial effects on
various performance measures. They do, however, assume that each vehicle in the
platoon has the knowledge of the desired (leader's) velocity of the platoon.
That information has to be communicated permanently to each vehicle by the
leader.
Our work to be presented differs in that we allow no communication among the vehicles.

The results in \cite{Tangerman2012} reveal a significant drawback of the
asymmetric control scheme. The paper analyzes a platoon of vehicles modeled by
double integrators with a~PD controller (equivalent to relative position and
velocity feedback). They show that the peak in the magnitude frequency response
of the position of the last vehicle to the change in the leader's position grows
exponentially in the number of vehicles---a phenomenon labelled as
\textit{harmonic instability}. In contrast, if the controller is symmetric, the
peak in the magnitude frequency response (the $\mathcal{H}_\infty$ system norm)
only grows linearly \cite{Veerman2007}. 
Note that string instability merely means that the
$\mathcal{H}_\infty$ norm is growing but harmonic instability means that it is
growing very fast (in the number of vehicles).

With these results, several questions arise. Is harmonic instability present
with any controller or can it be mitigated by some judicious choice of the
controller structure? Can varying the asymmetry in the platoon counteract
harmonic instability? Is harmonic instability an inherent property of an
asymmetric control, or even of any uniformly bounded nearest neighbor
interaction? In this paper we answer these questions.

We extend \cite{Tangerman2012} to any open-loop model of a vehicle and any
platoon with uniformly bounded eigenvalues. Our results also extend
\cite{Seiler2004a} from the predecessor-following architecture to any
bidirectional asymmetric configuration. Moreover, we show that harmonic
instability is, in fact, caused by the uniform boundedness and it is not
possible to achieve a good scalability both in the convergence time (the bound
on eigenvalues) and in the frequency domain (the $\mathcal{H}_\infty$ system
norm). Some trade-off is necessary. This paper extends our previous conference
paper \cite{Herman2014} to arbitrary asymmetric formations with controller gains
and asymmetries varying among the vehicles.

The paper is structured as follows. First we give some preliminaries and provide definition of the harmonic instability. Then we prove uniform boundedness of a general
platoon. In the next section the proof for the harmonic instability of an asymmetric control scheme is given. Finally some special cases are discussed and simulation results are shown.

\section{Preliminaries and model}
We assume $\numVeh$ vehicles indexed by $i=1,2, \ldots, \numVeh$, travelling in
a one-dimensional space. The first vehicle (indexed 1) is called
\textit{the~leader} and it is controlled independently of the rest of the
platoon. We analyze a bidirectional control, where each onboard controller
measures the distances to its immediate predecessor and follower and strives
to keep these close to the desired (reference) distance. It sets different
weights to the \annote[ZH]{front and rear}{Nahradil jsem misto puvodniho forward
and backward, protoze si nejsem jisty, zda backward neni spise zpetny ve smyslu
jdouci dozadu.} regulation errors, hence \textit{asymmetric bidirectional
control}. We assume no intervehicular communication; all information is obtained
only locally by the onboard sensors.

We study how the disturbance created by unexpected movements of the leader
propagates along the platoon towards the final vehicle. Hence, we analyze the
properties of the~transfer function
\annote[ZH]{$\diagTransBlockN(\lapDom)$}{Nechces to spise prejmenovat na $T_n$?
Stejne ten prvni ze dvou dolnich indexu k nicemu nepouzivas, a bude to
citelnejsi protoze strucnejsi.} from the leader's position to the position of
the last vehicle as depicted in Fig.~\ref{fig:totalTransferFunction}. Its
frequency response and the way it scales with the number of vehicles $\numVeh$
is used to prove harmonic instability for a given configuration.
\begin{definition}[Harmonic stability\cite{Tangerman2012}]
	Let \annote[ZH]{$\displaystyle  \gamma_\numVeh \equiv \mathrm{sup}_{\omega \in
	\mathbb{R}^+}|\diagTransBlockN(\jmath \omega)|$}{Navrhuji prejmenovat takto z puvodniho $A_N$. Myslim, ze to tu vice sedi. V teorii rizeni se casto pro Hinf normy pouziva, viz ``$\gamma$ iterations''. A naopak to $A_N$ je v rizeni strasne moc spojene s temi stavovymi maticemi... Dale, nejsem si jisty vhodnosti pouziti $\equiv$ pro ``je definovano jako''. To je asi beznejsi trebas proste $:=$. A nebo se na to vykasli uplne, a napis jen $=$, protoze stejne v tom rozlisovani, co je rovnice a co je definice, nejsi dusledny :-) Stejne jako ja a jako 99.9\% autoru.}, where $\jmath=\sqrt{-1}$. The platoon is called \textit{harmonically stable} if it is asymptotically stable and if $ \displaystyle \lim\!\sup_{\numVeh \to \infty}
	\gamma_\numVeh^{1/N} \leq 1$. Otherwise it is \textit{harmonically unstable}.
\end{definition}
An interpretation of harmonic instability is that some oscillatory motion of the
leader has its amplitude magnified as it is propagated through the platoon and
the growth of the magnitude is exponential in $\numVeh$ \cite{Tangerman2012}.
\annote[ZH]{In order words,}{Tady mi to skutecne prislo, ze je to jen ta stejna
vec recena jinymi slovy, to neni nejaky dalsi vysledek ci dalsi dopad.} the
$\mathcal{H}_\infty$ norm of $\diagTransBlockN(s)$ grows exponentially with
$\numVeh$.
\begin{figure}
\centering
	\includegraphics[width=0.4\textwidth]{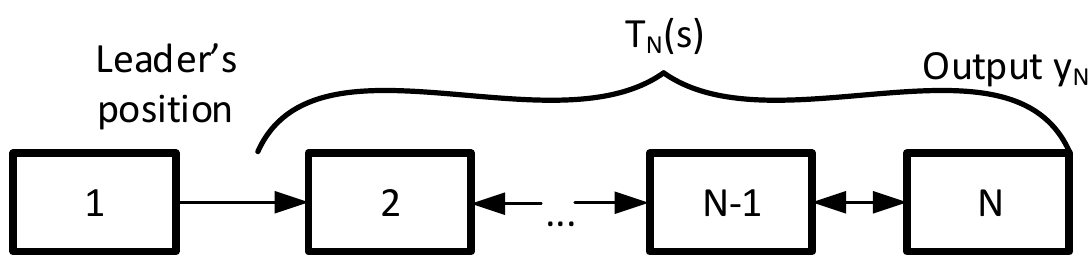}
	\caption{Transfer function from the leader's position (second vehicle's input)
	to the trailing vehicle's output.}
	\label{fig:totalTransferFunction}  
\end{figure}


\subsubsection*{Notation}
We denote \annote[ZH]{matrices by capital letters}{Vidis, i proto jsem chtel mit
to skalarni cislo reckym pismenem. Zkusme se toho v budoucnu drzet, kdyz to jen
trosku pujde. Jasne, ne vzdy to pujde.}, vectors by lowercase letters and an
element in a matrix $A$ is denoted as $a_{ij}$. We use $\lapDom$ as the complex
variable in Laplace transform. The $\wn$th vector in a canonical basis is
denoted $e_i \in \mathbb{R}^{\numVeh \times 1}$, that is, $e_\wn=[0, \ldots,
1,\ldots,0]^T$, with 1 on the $\wn$th position. Identity matrix of size
$\numVeh$ is denoted as $\Id_\numVeh$.

\subsection{System model}
Each vehicle is described by an identical SISO transfer function
$\vehicleTf(\lapDom) = \frac{\vehNumCoef(\lapDom)}{\vehDenCoef({\lapDom})}$. The
\annote[ZH]{output is}{Tedy kdyz uz je receno, co je vystupem, tak zvedava duse
ctenare jiste bude touzit po poznani fyzikalni povahy vstupu...} the vehicle's
position $\pos_i$. Dynamic controller described by a transfer function
$\controllerTf(\lapDom) = \frac{\contNumCoef(\lapDom)}{\contDenCoef(\lapDom)}$
is used to close the feedback loop. The input to the controller is defined in
(\ref{eq:regError}). The open-loop transfer function is
$\openLoop(\lapDom)=\controllerTf(\lapDom)\vehicleTf(\lapDom)$. From now on we
will only use the open loop in the analysis. Its state-space description is
\begin{equation}
	\dot{\stateVect}_i = \matAi \stateVect_i + \matBi \totalInput_i, \qquad  \outVect_i = \matCi \stateVect_i,
\end{equation}
with $\stateVect_i \in \mathbb{R}^{\olOrd \times 1}$ as the state vector,
$\matAi \in \mathbb{R}^{\olOrd\times \olOrd}, \matBi \in \mathbb{R}^{\olOrd \times 1}, \matCi \in \mathbb{R}^{1 \times \olOrd}$, and $\pos_i \in \mathbb{R}$ as the position of the vehicle.

The \annote[ZH]{input has two parts}{Priznam se, ze toto je pro me docela 
matouci. Vice v poznamkach.} $\totalInput_i=\relCoupling_i+\inp_i$, where
$\relCoupling_i$ is the part caused by coupling between vehicles and $\inp_i$
external control signal, e.g., reference distances. 

In platooning, the input $\totalInput_i$ to each vehicle is a weighted sum of
spacing errors to its predecessor in the string and its successor. Spacing error
to the previous vehicle is weighted with $\weightGain_i>0$ and the error to
succeeding vehicle with $\weightGain_i \epsilon_i$. The asymmetry level
$\epsilon_i \geq 0$ is therefore a ratio between front and rear gains. The
weight $\weightGain_i$ as well as the asymmetry level $\epsilon_i$ can vary
along the platoon. The input to the vehicle is then
\begin{equation}
	\totalInput_i = \weightGain_i(\pos_{i-1}-\pos_{i}-\dreference) -
	\weightGain_i \epsilon_i(\pos_i- \pos_{i+1}-\dreference)
	\label{eq:regError}
\end{equation}
for $i=2, \ldots, \numVeh$ and \annote[ZH]{$\dreference$ is a desired
distance}{To jsem Ti rikal, ze Bassamovi studenti prezdivaji ``notation Nazi''?
:-) Taky mam poznamku k tomuto v prilozenem textu}. The inter-vehicle coupling
is then $\relCoupling_i = \weightGain_i(\pos_{i-1}-\pos_{i}) - \weightGain_i
\epsilon_i(\pos_i-\pos_{i+1})$ and the external input in this case is
$\inp_i=\weightGain_i(-1+\epsilon_i)\dreference$. \annote[ZH]{In further
development, we do not limit the external command $\inp_i$ to be given only by
the reference distance, it is treated as a general signal.}{Tomu uplne
nerozumim.
Jak pak budu moci interpretovat (2), kdyz d, tedy zatim delta, bude neco jineho nez
referencni distance? Nevyhodime tuto poznamku radsi?}\annote[IH]{}{Ona to prave
referencni vzdalenost neni. Pro dvojsmerne rizeni tam bude
$\dreference-\dreference=0$. Ani to $(-1+\epsilon)$ neni referencni vzdalenost,
ale nejaka funkce. A odezva na skok by pak byla odezva na zmenu referencni
vzdalenosti jen u prvniho. Je to skutecne libolny signal. Uz jen delat
sinusovou zmenu reference jen dopredu pro dukaz harmonicke nestability je
divne.}
	
In a compact form, the stacked vector of inputs is $\totalInput = -\lapl \systemOutputVect + \inpVect$, where $\lapl$ is a graph Laplacian representing an interconnection and has a form 
\vspace{-5pt}

{
\setlength{\arraycolsep}{2pt}
\begin{equation}
\lapl= \left[ \begin{matrix}
							0 & 0 & 0 & 0 & \ldots & 0 \\
							-\weightGain_2 & \weightGain_2(1+\epsilon_2) & -\weightGain_2 \epsilon_2
							& 0 & \ldots & 0
							\\
							0 & -\weightGain_3 & \weightGain_3(1+\epsilon_3) &
							-\weightGain_3\epsilon_3 & \ldots & 0
							\\
							\vdots & \vdots & \vdots & \vdots & \ddots & \vdots \\
							0 & 0 & \ldots &  0 & -\weightGain_\numVeh & \weightGain_\numVeh
 							\end{matrix} 
							\right].
							\label{eq:laplacian}
\end{equation}
}
\setlength{\arraycolsep}{3pt}

The vectors are $\totalInput=[\totalInput_1, \ldots,\totalInput_\numVeh]^T$,
$\systemOutputVect=[\pos_1, \ldots,\pos_\numVeh]^T$ and since the external input
can be arbitrary, we take $\inpVect=[\inp_1,\ldots,\inp_\numVeh]^T$. The leader
is controlled externally with the input $u_1$. The trailing vehicle has no
follower and its input is $\totalInput_\numVeh = \weightGain_\numVeh(\pos_{\numVeh-1}-\pos_{\numVeh}-\dreference)$.

\begin{lemma}
	The graph Laplacian in (\ref{eq:laplacian}) has these properties
	\begin{enumerate}
	  \item[a)] $\spatEig_1=0$ is an eigenvalue of $\lapl$ with its eigenvector
	  $\mathbf{1}=[1,1,\ldots,1]^T$, and this eigenvalue is simple,
	  \item[b)] all its eigenvalues are real, lie in the right-half plane, i.e.
	  $\spatEigWn \geq 0$ and are bounded by $\spatEig_i \leq \spatEigMax = 2\max(\lapl_{ii})$.
	  \item[c)] $\lapl$ can be partitioned as
	  \begin{equation}	
	 \lapl = \left[\begin{array}{c|c}
	                0 & 0\\
	                \hline	    
	                \times & {\redLapl}
	               \end{array}\right].
	               \label{eq:decompLapl}
	\end{equation}
	and the spectrum of the reduced Laplacian $\redLapl$ coincides with all the
	nonzero eigenvalues of $\lapl$.
	\end{enumerate}
	\label{lem:lapProp}
\end{lemma} 
\begin{proof}
a) Simple zero eigenvalue follows from the presence of a directed spanning tree
in the platoon. b) $\lapl$ \annote[ZH]{is a tridiagonal real matrix}{No tak
hlavne je to ale Laplacian, ne? A pak se asi k realnosti vlastnich cisel ani
netreba vyjadrovat. At nerikaji, ze objevujeme
Ameriku.}\annote[IH]{}{Realna vlastni cisla rozhodne nepatri mezi zname a
bezne veci.
Zname je to pro symetricke matice, ale L symetricka matice zdaleka neni.
Vsadim se, 8 z 10 lidi v rizeni nebude znat, ze tridiagonalni matice tohoto
typu maji jen realna vlastni cisla.
Obecne ma Laplacian vlastni cisla komplexni, coz prave cini navrh rizeni pro
formaci tak tezkym.} with non-positive off-diagonal terms, so its eigenvalues
are real \cite[Lem.
0.1.1]{Fallat2011}.
The upper bound follows from Ger\v{s}gorin's theorem \cite[Thm.
6.1.1]{Horn1996}. c) Combining the property that the first row of $\lapl$ is
zero and one of the eigenvalues is zero, similarity transformation reveals the
eigenstructure described in the lemma.
\end{proof}	
	Using the last point, we can concentrate on the formation without the
	leader, because we removed the row corresponding to the leader and still kept all
	the nonzero eigenvalues. That's why the input to $\diagTransBlockN(s)$ acts at
	the second vehicle.
\begin{definition}[Uniform boundedness] The eigenvalues $\spatEig_i$ of a matrix $\lapl \in \mathbb{R}^{\numVeh \times \numVeh}$ are uniformly bounded from zero if there exists a constant $\spatEigMin>0$ such that $\spatEig_i  \geq \spatEigMin$ for  $i=2,\ldots,\numVeh$ and $\spatEigMin$ does not depend on $\numVeh$.
\end{definition}
If the onboard controllers of all vehicles are asymmetric and have
front gains stronger than the rear ones, then the uniform boundedness can be
achieved.
The proof of the following theorem is in the Appendix A.
\begin{theorem}
If there is $\epsilon_{\max} < 1$ such that $\epsilon_i \leq
\epsilon_{\max} \, \forall i$ and $\forall \numVeh$, then the nonzero
eigenvalues of the Laplacian $\lapl$ given in (\ref{eq:laplacian}) are uniformly bounded with 
$\spatEigMin \geq \frac{(1-\epsilon_{\max})^2}{2+2\epsilon_{\max}}$.
	\label{thm:uniformBound}
\end{theorem}

\subsection{Vehicle interconnection and diagonalization} 
Using a standard consensus or multi-vehicular formation notation
\cite{Fax2004a}, the overall formation model is 
\begin{IEEEeqnarray}{rCl}
	\stateVectDer &=& \bigg[ \IN \otimes \matAi - (\IN \otimes \matBi \matCi)(\lapl
	\otimes \In)\bigg] \stateVect + (\IN \otimes \matBi) \inpVect,
	\nonumber \label{eq:overallStateEq}\\
	\outVect &=& (\IN \otimes \matCi) \stateVect \label{eq:overallMeasEq},
\end{IEEEeqnarray}
where $\stateVect \in \mathbb{R}^{\numVeh \olOrd \times 1}$ is a
\textit{stacked} state vector and $\otimes$ is the Kronecker product. We apply the approach of Theorem
1 from \cite{Fax2004a} but use Jordan instead of Schur decomposition. This will
block diagonalize the system. The state transformation is
$\stateVect=(\matEigVect \otimes \In)\stateVectMod$, where $\matJ =
\matEigVect^{-1}\lapl \matEigVect$ is the Jordan form of $\lapl$.
The matrix $\matEigVect=[\eigVect_1, \ldots, \eigVect_\numVeh]$ is formed by
(generalized) eigenvectors of $\lapl$ and $\eigVect_{ji}$ is the $j$th element
of the vector $\eigVect_i$. A block diagonal system is
\begin{IEEEeqnarray}{rCl}
\stateVectModDer &=& \left[\IN \otimes \matAi - \matJ \otimes \matBi \matCi
\right] \stateVectMod + (\matEigVect^{-1} \otimes \matBi) \inpVect,
\label{eq:overallSystemEqMod}
\\
\outVect &=& (\matEigVect \otimes \matCi) \stateVectMod \label{eq:overallOutputEqMod}.
\end{IEEEeqnarray}
Consider a Jordan block in the block diagonal matrix (\ref{eq:overallSystemEqMod}). If it
is of size one, it has the form
\begin{equation}
\stateVectModDer_\wn = \left[\matAi - \spatEigWn \matBi \matCi \right]
\stateVectMod_\wn +  \matBi e_\wn^T \matEigVect^{-1} \inpVect, \quad
\hat{\outVect}_i=\matC \stateVectMod_i
\label{eq:diagBlockStateSpace}
\end{equation}
This equation can be viewed as an output feedback system with a feedback
gain $\spatEigWn$ and output $\hat{\pos}_i$. Its transfer function is
\begin{equation}
	\diagTransBlockWn(\lapDom)=
	\frac{\openLoop(\lapDom)}{1+\spatEigWn
	\openLoop(\lapDom)}=\frac{\vehNumCoef(\lapDom)\contNumCoef(\lapDom)}{\vehDenCoef(\lapDom)\contDenCoef(\lapDom)
	+ \spatEigWn \vehNumCoef(\lapDom)\contNumCoef(\lapDom)}.
\end{equation}
If the Jordan block has a size larger than one, it corresponds to identical blocks connected in series as in Fig.~\ref{fig:diagonalizedSystemJordan}. 

In the following we assume that all diagonal blocks are asymptotically stable for all $\numVeh$. As all the eigenvalues $\spatEigWn$ are real, design of a stable system is not a difficult task. We can use, e.g., the synchronization region approach \cite{Zhang2011} or the root-locus-like approach \cite{Herman2013}.
\begin{figure}
\centering
	\include{figures/ClosedLoopSeriesSmall}
\caption{\annote[ZH]{Block diagram for a Jordan block of size two.}{No, on az
tak moc diagonalni neni, ne? Asi spise jednoduse``''?A diagonal block for the
Jordan block of size 2.} The \annote[ZH]{input}{Kdyby ses rozhodl to preznacit v
textu (misto $u$ neco jineho, trebas prave $d$ nebo $r$ nebo $w$), tak nezapomen
i tady v obrazku.} is applied to the second vehicle and the output is the
position of the $\numVeh$th vehicle.
}
	\label{fig:diagonalizedSystemJordan}    
\end{figure}

\section{Harmonic instability} 
To test for harmonic instability, we examine the transfer function
$\diagTransBlockN(\lapDom)$. 
The input to the platoon for such transfer function is
$\inpVect(\lapDom)=[0,\inp_2(s), 0, \ldots, 0]^T=e_2 \inp_2(\lapDom)$.
Based on
(\ref{eq:diagBlockStateSpace}), the input to the diagonal block
$\diagTransBlockWn(\lapDom)$ is the $\wn$th entry in the vector $\inpDiag$ given
by $\inpDiag(\lapDom) = \matEigVect^{-1} e_2 \inp_2(\lapDom) = g
\inp_2(\lapDom)$ with $g=\matEigVect^{-1} e_2$. The output of each block is (see Fig.
\ref{fig:diagonalizedSystemJordan})
	\begin{equation}
		\hat{\pos}_i(\lapDom)=\diagTransBlock_i(s)\inpDiag_i(s) =
		\diagTransBlock_i(s) g_i \inp_2(\lapDom).
	\end{equation}
The position $\pos_\numVeh$ of the $\numVeh$th vehicle can be calculated from
(\ref{eq:overallOutputEqMod}).
It is a weighted sum of the outputs of the blocks $\hat{\pos}$ with the weights
equal to the $\numVeh$th terms in the eigenvectors $\eigVect_{i}$
\begin{IEEEeqnarray}{rCl}
	\pos_\numVeh (\lapDom) &=&
	\sum_{\wn=1}^{\numVeh} \eigVect_{\numVeh \wn}
	\hat{\pos}_i(s)=\left[\sum_{\wn=1}^{\numVeh}
	\eigVect_{\numVeh \wn} \diagTransBlock_\wn(s) g_i
	\right]\inp_2(s), \label{eq:transferFunctionNSum}
\end{IEEEeqnarray}
with which we define
$\diagTransBlockN(s)=\frac{\pos_\numVeh(s)}{\inp_2(s)}=\sum_{\wn=1}^{\numVeh}
\eigVect_{\numVeh \wn} \diagTransBlock_\wn(s) g_i$.
\annote[ZH]{}{Asi bych
ten nazev podsekce dal pryc, vzdyt je to jen jedina podsekce v ramci sekce, takove strukturovani nema cenu. A usetrime docela dost mista. Smazes?}
\note[IH]{Smazano}{}
 
The following product form of $\diagTransBlockN(s)$ holds for general platoons
(both symmetric and asymmetric).
\begin{theorem}\note[ZH]{Ja bych spise Lemma 3 prejmenoval na Theorem. Vzdyt to
je prece silny vysledek. Jeden z hlavnich v tomto clanku, ne? To, ze muzu popis kolony vyjadrit v te seriove verzi. Pak bude i opravnene pouziti toho nasledujici Corollary, protoze ten se skutecne spise pouziva pro dodatecne vysledky bezprostredne plynouci z Theoremu.} The transfer function from the input of the second vehicle to the
	position of the last vehicle in the system (\ref{eq:overallStateEq}) with
	Laplacian (\ref{eq:laplacian}) is given as
	\begin{IEEEeqnarray}{rCl}
		\diagTransBlockN(\lapDom) &=& \frac{1}{\weightGain_2} \prod_{\wn=2}^{\numVeh}
		\spatEigWn \diagTransBlockWn(\lapDom) =\frac{1}{\weightGain_2} \prod_{\wn=2}^{\numVeh}
		\diagTransBlockEigWn(\lapDom) 
		,
		\label{eq:blockProd}
	\end{IEEEeqnarray}
	\label{lem:blockProd} 
\end{theorem} 
\vspace{-15pt}

We introduced
$\diagTransBlockEigWn(\lapDom)=\spatEigWn\diagTransBlockWn(\lapDom)$ as a
closed-loop transfer function with gain $\spatEig_i$. The proof is given in the
Appendix B.
\begin{corollary} 
	For at least one integrator in the open loop $\openLoop(s)$,
	the steady-state gain of each block in (\ref{eq:blockProd})
	is $\diagTransBlockEigWn(0)=1, \, \forall i$ and then $\diagTransBlockN(0)=1$
	if and only if $\weightGain_2=1$.
\end{corollary}
Surprisingly, the greater the gain $\weightGain_2$ (coupling with the
leader), the lower the steady-state gain of the platoon.
 
Before stating the main theorem of the paper, we need to introduce the notation
$\openLoopGain(s)=\spatEigMin \openLoop(s)$  and define the closed-loop block
for such open loop as
\begin{equation}
	\diagTransBlockEigMin(s)=
	\frac{\openLoopGain(s)}{1+\openLoopGain(s)}=\frac{\spatEigMin
	\vehNumCoef(s)\contNumCoef(s)}{\vehDenCoef(s)\contDenCoef(s) + \spatEigMin \vehNumCoef(s)\contNumCoef(s)}. \label{eq:diagBlockEigMin}
\end{equation}

\begin{theorem}
If the nonzero eigenvalues of Laplacian in (\ref{eq:laplacian}) are
uniformly bounded and $||
\diagTransBlockEigMin(s) ||_{\infty} > 1$, then the platoon is harmonically
unstable.
\label{thm:stringInstab}
\end{theorem}
 \begin{proof} The condition states that closed-loop block
 $\diagTransBlockEigMin(s)$ defined in (\ref{eq:diagBlockEigMin}) corresponding 
 to the lower bound on the eigenvalues of Laplacian is greater in the
 $\mathcal{H}_{\infty}$ norm than one. Let $\omega_0$ be the frequency at
 which the magnitude frequency response of this block \annote[ZH]{attains its
 maximum}{Technicka pripominka: ona zadna takova frekvence existovat nemusi, ne?
 I proto je Hinf norma definovana nikoliv jako maximum nybrz jako suppremum. A
 nebude toto mit dokonce dopad na dukaz?}. Further let $\alpha+\jmath \beta$
 (with $\sqrt{-1}=\jmath$) be the value of the frequency response of the scaled
 open loop $\openLoopGain$ at $\omega_0$, i.~e., $\openLoopGain(\jmath
 \omega_0)= \alpha + \jmath \beta$. Then the squared modulus of the frequency
 response of the closed-loop $\diagTransBlockEigMin(\lapDom)$ \annote[ZH]{reads}{Pozor,
 to na prave strane je ve skutecnosti kvadrat te leve, ne? Vice v poznamkach.}
\begin{equation}
	|\diagTransBlockEigMin(\jmath \omega_0)|^2 = \left |
	\frac{{\openLoopGain}(\jmath \omega_0)}{1+{\openLoopGain}(\jmath
	\omega_0)} \right |^2=\frac{\alpha^2 + \beta^2}{(\alpha+1)^2 + \beta^2}.
	\label{eq:modulMin}
\end{equation}
Since at $\omega_0$ the closed-loop magnitude frequency response attains its
maximum, the peak is greater than 1, i. e., $|\diagTransBlockEigMin(0)|=1 <
|\diagTransBlockEigMin(\jmath \omega_0)|$. From (\ref{eq:modulMin}) we have
\begin{equation}
	\frac{\alpha^2 +\beta^2}{(\alpha+1)^2 + \beta^2} > 1 \Rightarrow \alpha <
	-\frac{1}{2}.
\end{equation} 
	
The Laplacian eigenvalues can be ordered as $\spatEigMin \leq \spatEig_2 < \ldots \leq \spatEigMax$. We can write $\spatEigWn = \kappa_\wn \spatEigMin$ with gain $\kappa_\wn \in \langle 1,\frac{\spatEigMax}{\spatEigMin} \rangle$. By Lemma \ref{lem:lapProp} all eigenvalues are real, so $\kappa$ is real as well.
By assumption in the theorem the bounds on $\kappa$ do not depend on the number
of vehicles.

Now the transfer function of each term in the product (\ref{eq:blockProd}) is
$\diagTransBlockEigWn(\lapDom) = \frac{\kappa_\wn
\openLoopGain(\lapDom)}{1+\kappa_\wn \openLoopGain(\lapDom)}$ with
$\kappa_\wn=\frac{\spatEigWn}{\spatEigMin}$. We prove that all such
transfer functions also have the magnitude frequency response at $\omega_0$
greater than 1 (not necessarily their maximum there). The value of
$\openLoopGain(\jmath \omega_0)$ is still written as $\alpha+\jmath \beta$. The
\annote[ZH]{squared modulus}{Jen si i tady zkontroluj, jestli to neni ve
skutecnosti ten kvadrat. Mne se to nechtelo kontrolovat.} of the closed-loop frequency response
at $\omega_0$ is
\begin{equation}
	|\diagTransBlockEigWn(\jmath \omega_0)|^2 = \left |
\frac{\kappa_\wn {\openLoopGain}(\jmath \omega_0)}{1+\kappa_\wn
\openLoopGain(\jmath \omega_0)} \right |^2=1- \frac{2 \kappa_\wn
\alpha+1}{(\kappa_\wn \alpha+1)^2 +\kappa_\wn^2 \beta^2}. \label{eq:modulK}
\end{equation}
Since $\alpha < -\frac{1}{2}$, $\kappa_\wn$  is real and greater than $1$ and
the denominator is positive, the sign of the fraction must be negative and
(\ref{eq:modulK}) is greater than 1. Therefore, all transfer functions
$\diagTransBlockEigWn(\lapDom)$ at $\omega_0$ have the modulus greater than 1.
	
The modulus of the frequency response parametrized by $\kappa$ attains its
minimum at $\omega_0$ for some $\kappa_0$, independent of the number of
vehicles. This smallest modulus at $\omega_0$ is denoted as
\annote[ZH]{$\zeta_{\min}>1$}{Dovolil jsem si i tady nahradil velke pismeno $H$
radsi reckym pismenem. Tady to $H_\mathbf{min}$ navic i trosku kolidovalo s
$H_\infty$.} and it is unchanged for any and all diagonal blocks. By Theorem
\ref{lem:blockProd}, the diagonal blocks are connected in series, therefore the
total gain of the platoon is given by product
\begin{equation}
|\diagTransBlockN (\jmath \omega_0)| =
\prod_{i=2}^{\numVeh}|\diagTransBlockEigWn(\jmath \omega_0)| \geq
(\zeta_{\min})^{N-1}.
\end{equation}
The exponential growth of the peak in the magnitude frequency response has thus
been proved. Although the eigenvalues of the Laplacian change upon adding more
vehicles into the platoon, the bound on eigenvalues as well as the corresponding
gain $\zeta_{\min}$ remain constant.
\end{proof} 

To summarize, it suffices to test only a single transfer function
$\diagTransBlockEigMin(s)$ instead of the model of the whole platoon. If this
transfer function is larger in $\mathcal{H}_\infty$ than one and there is a
lower bound on the Fiedler's eigenvalue, the harmonic instability must occur and
cannot be overcome by any linear controller. Note, however, that even systems
with only one integrator in the open loop can be harmonically unstable.
\section{Special cases and simulations}
A particularly important case is when there are two integrators in the open
loop.
\begin{lemma}
	For at least two integrators in the open loop, frequency response of each term 	in the product (\ref{eq:blockProd}) has a resonance peak, i. e. $|\diagTransBlockEigWn(0)|=1 < |\diagTransBlockEigWn(\jmath \omega_\wn)|$ for some $\omega_\wn$.
	\label{lem:overshoot}
\end{lemma}
\begin{proof}
	Each term in the product in Theorem \ref{lem:blockProd} is a closed-loop
	transfer function with at least two integrators in the open loop. For such system it was proved in Theorem 1 in \cite{Seiler2004a} that it must have $H_\infty$ norm greater than 1. 
\end{proof} 
Using the fact that $||\diagTransBlockEigMin(s)||_\infty > 1$ with at least two
integrators in the open loop, we satisfy the
conditions in Theorem \ref{thm:stringInstab} and can extend the results of
\cite{Seiler2004a}.
\begin{corollary}
	Vehicular platoon with uniformly bounded eigenvalues of Laplacian and at least	two integrators in the open loop is harmonically unstable. This cannot be 	cancelled by any linear controller.
\end{corollary}

Theorem \ref{thm:uniformBound} proves uniform bound for arbitrary asymmetric formation, so we can extend results of \cite{Tangerman2012} to varying asymmetry
and arbitrary dynamical models with two integrators.
\begin{corollary}
	Asymmetric bidirectional control with $\epsilon_i \leq \epsilon_{\max} < 1 \,
	\forall i, \, \forall \numVeh$ and with at least two integrators in the open
	loop is harmonically unstable.
	\label{cor:asymHarmUnstable}
\end{corollary}
It was proved in \cite{Herman2013,Tangerman2012} that if the asymmetric platoon
uses identical asymmetries $\epsilon_i=\epsilon$, $\mu_i=1$, the eigenvalues of
Laplacian are given in closed form as $\spatEigWn = -2 \sqrt{\epsilon} \cos
\spatFreqWn + 1 + \epsilon$, where $\spatFreqWn$ is given as the $\wn$th
solution of the nonlinear equation $\sin(N\spatFreqWn) -
\sqrt{\frac{1}{\epsilon}} \sin \big((N+1) \spatFreqWn \big)$ on the interval
$\langle0, \pi\rangle$.  The Laplacian eigenvalues are thus bounded and the bounds
$\spatEigMin \geq (1 - \sqrt{\epsilon})^2, \spatEigMax \leq (1 + \sqrt{\epsilon})^2$ do not depend on $\numVeh$. Such formation satisfies the
conditions of Corollary \ref{cor:asymHarmUnstable}.

Another special case, which is harmonically unstable, is the predecessor following algorithm with $\epsilon_i=0$. On the other hand, harmonic stability
of symmetric bidirectional control ($\epsilon_i=1$) for double integrator model was proved in \cite{Veerman2007}. 
 \begin{figure}[t] 
\centering
	\begin{subfigure}[b]{0.35\textwidth}
	\includegraphics[width=1\textwidth]{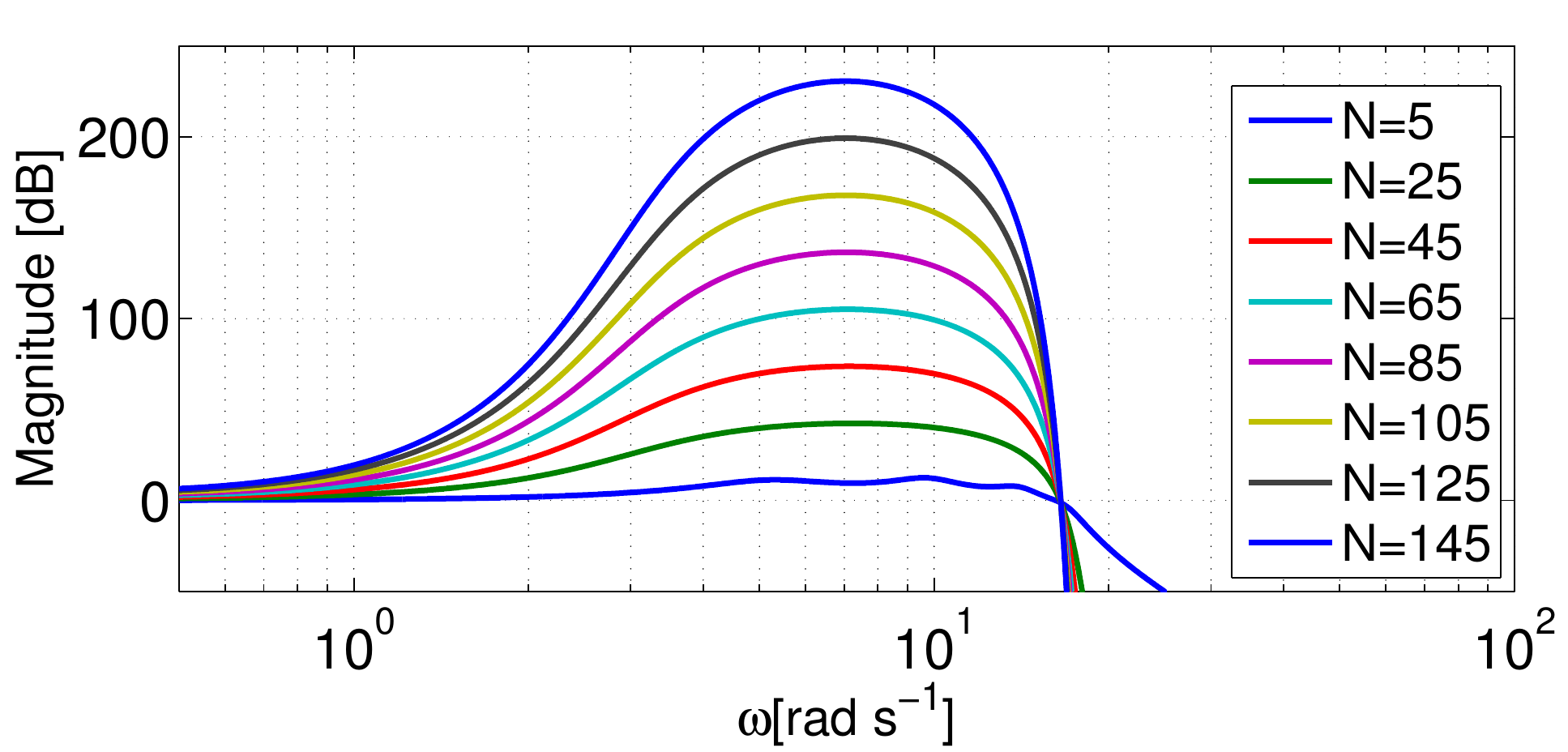}
	\caption{Freq. response of asymmetric control}
	\label{fig:stringInstab}
	\end{subfigure}   
	\begin{subfigure}[b]{0.35\textwidth} 
	\includegraphics[width=1\textwidth]{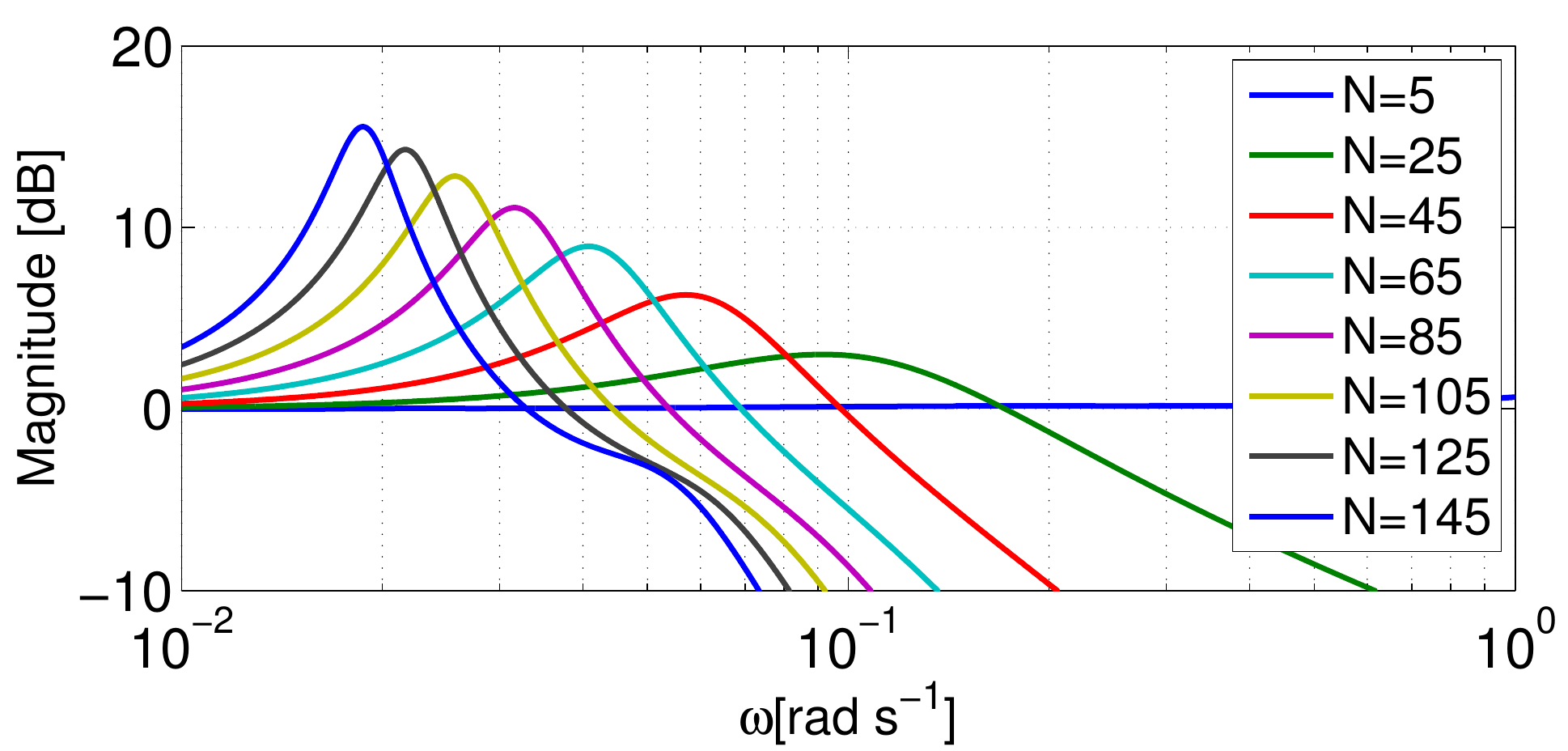}	
	\caption{Freq. response of symmetric control}\label{fig:stringStab}  
	\end{subfigure} 
	\begin{subfigure}[b]{0.35\textwidth}
	\includegraphics[width=1\textwidth]{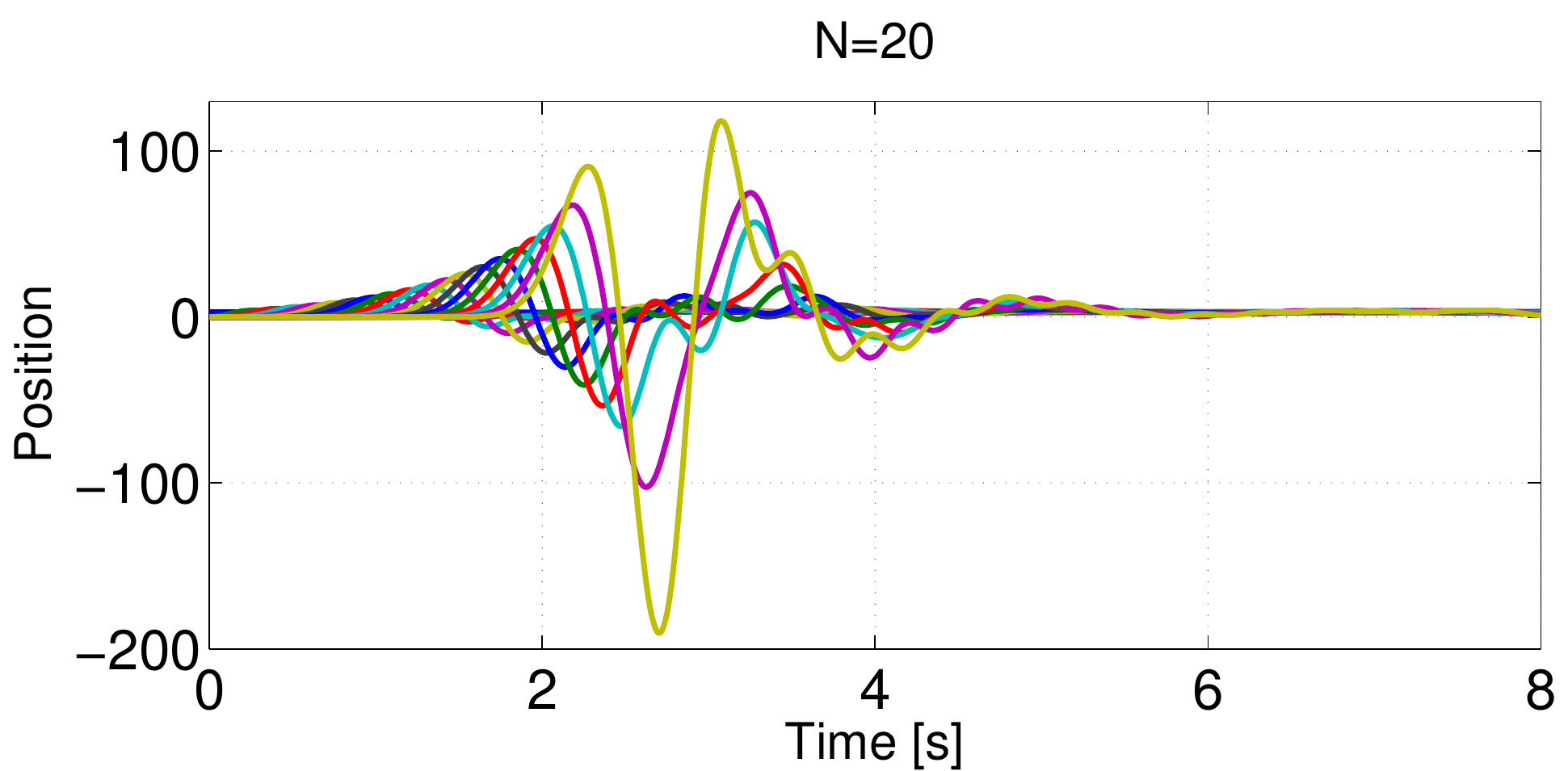}	
	\caption{Response to the leader's step in position.}\label{fig:stepResp}
	\end{subfigure}
	\caption{Figures \ref{fig:stringInstab} and \ref{fig:stringStab} show frequency
	responses for a vehicular platoon with a growing number of vehicles. Figure \ref{fig:stepResp} shows a response to leader's unit step change
	in position for an asymmetric platoon with 20 vehicles, $\epsilon=0.5$.}
	\vspace{-10pt}
\end{figure} 

The simulation results comparing the asymmetric control with $\epsilon_i=0.5$ and the symmetric control with $\epsilon_i=1$ are shown in Fig.~\ref{fig:stringInstab} and \ref{fig:stringStab}. For the asymmetric control scheme it is apparent that as the number of vehicles grows, the peak in the magnitude frequency response
grows exponentially and it is localized at almost identical and wide frequency
range for any number of cars. The figure \ref{fig:stepResp} shows step
response of the platoon, which is oscillating and has very high overshot. The
models used in all cases are $\controllerTf(\lapDom)=\frac{110 s^2 + 43s + 3}{s^2 + 2.9s+ 1}$, $\vehicleTf(\lapDom)=\frac{1}{s^2}$. The controller has been chosen so that the overall system is asymptotically stable for any number of vehicles.
  
\section{Conclusion}
We dealt with a vehicular platoon controlled in a distributed and asymmetric way
where each vehicle only measures the distance to its immediate neighbors. 
We studied \textit{harmonic instability} of
the platoon, which is a term 
for \textit{exponential scaling}
of the $\mathcal{H}_\infty$ norm of the transfer function of the platoon as the
number of vehicles in the platoon grows. 

The key condition for harmonic
instability is the uniform boundedness of the Laplacian eigenvalues. For
platoons with uniform boundedness we proposed a simple
test consisting of evaluation of the
$\mathcal{H}_{\infty}$ norm of a closed-loop transfer function of a single
vehicle.
The proof is based on a \textit{product form} of the transfer function
from the input of the second vehicle to the position of the last vehicle. 

We proved uniform boundedness for platoon with stronger front gains. In the case
of two or more integrators in the open-loop transfer function and uniform bound
on eigenvalues, harmonic instability cannot be overcome by any linear
controller.
The benefits of a uniform boundedness are thus
paid for by a very bad scaling in the frequency response.
Nonetheless, harmonic instability can also occur even in a situation with a single integrator in the
open-loop model.


\appendices
\section{Proof of Theorem \ref{thm:uniformBound}}
Before we proceed to the proof, we state one useful Lemma.
\begin{lemma}{\cite[Cor. 6.1.6]{Horn1996}}
	Let $A=[a_{ij}] \in \mathbb{R}^{n\times n}$ and let $p_1, \ldots, p_n$ be positive numbers. Consider the matrix $B=P^{-1}A P$ with $P = \mathrm{diag}(p_1, \ldots, p_n)$ and $b_{ij}=[p_j a_{ij}/p_i]$. Then all eigenvalues of $A$ lie in the union of Ger\v{s}gorin disks
	\begin{equation}
		\bigcup_{i=1}^{n}\left\{ z \in \mathbb{C}: |z-a_{ii}| \leq\frac{1}{p_i} \sum_{j=1, j \neq i}^{n} p_j |a_{ij}|  \right\}
	\end{equation}
	\label{cor:diagTransf}
\end{lemma}	
\begin{proof}[\textbf{Proof of Theorem \ref{thm:uniformBound}}]
With Lemma \ref{cor:diagTransf} we can get tighter bounds on $\spatEigWn$ by transforming the reduced Laplacian $\redLapl$ into a \textit{diagonally dominant} form $B=P^{-1}\redLapl P$. After the transformation, each row of $B$ reads
\begin{equation}
	\left[
		\begin{matrix}
			 \ldots \!&\! 0 \!&\! -\frac{p_{i-1}}{p_i}\weightGain_i & \weightGain_i(1+\epsilon_i) & -\frac{p_{i+1}}{p_i}\weightGain_i \epsilon_i \!&\! 0\!&\! \ldots
		\end{matrix} 
	\right].
\end{equation}
To make it diagonally dominant, it must hold
\begin{equation}
	-\frac{p_{i-1}}{p_i}\weightGain_i + \weightGain_i(1+\epsilon_i) - \frac{p_{i+1}}{p_i}\weightGain_i \epsilon_i \geq 0 \quad \forall i. \label{eq:difEq}
\end{equation}
This is a difference inequality with variable $p$. 
We take $p$ as
\begin{equation}
	p = \frac{1}{2}\left(1+\frac{1}{\epsilon_{\max}}\right), \label{eq:p}
\end{equation}
which satisfies the inequality. Then $P$ is a diagonal matrix $P=\mathrm{diag}(1, p, p^2, \ldots,
p^{\numVeh-2})$. Applying this transformation to $\redLapl$, we get the $i$th row
	\begin{equation}
		\left[
			\begin{matrix}
				\ldots\!&\!0 \!& -\frac{1}{p}\weightGain_i & \weightGain_i(1+\epsilon_i) & -p\weightGain_i \epsilon_i & 0\!&\!\ldots
			\end{matrix} 
		\right]. \label{eq:rowP}
	\end{equation}	
The sum in each row equals the distance $d_i = \weightGain_i(1+\epsilon_i) - \frac{1}{p}\weightGain_i -p \weightGain_i\epsilon_i \label{eq:rowSum}$ of Ger\v{s}gorin's circle from zero and should be positive. After simple calculations, we obtain
\begin{equation}
	d_i = \weightGain_i \left[ -\frac{\epsilon_i}{2} \frac{1-\epsilon_{\max}}{\epsilon_{\max}}  + \frac{1-\epsilon_{\max}}{1+\epsilon_{\max}}\right].
\end{equation}
Assume, without loss of generality, that $\weightGain_i \geq 1$. Then $d_i$ in the equation above is minimized for $\epsilon_i = \epsilon_{\max}$. Therefore, the smallest distance of Ger\v{s}gorin disks from zero, hence also the lower bound on the eigenvalues is 
\begin{equation}
	\spatEig_{\min} \geq -\frac{1-\epsilon_{\max}}{2}  +
	\frac{1-\epsilon_{\max}}{1+\epsilon_{\max}} = \frac{(1-\epsilon_{\max})^2}{2+2\epsilon_{\max}}.
\end{equation}
Furthermore, it is positive for any $\epsilon_i \leq \epsilon_{\max}$, making $B$ diagonally dominant. To summarize, we found a bound which does not depend on the matrix size.
\end{proof}

\section{Proof of Theorem \ref{lem:blockProd}} 
Before the proof, we need the following technical result.
\begin{lemma}
Let $h_\wn=g_\wn  \,	\eigVect_{\numVeh \wn}$. Then we have
\begin{IEEEeqnarray}{rCl}
	&&\sum_{\wn=2}^{\numVeh} h_\wn \spatEigWn^m = 0 \text{ for }
	m = 0,1,\ldots, \numVeh-3, \label{eq:sumHiL}
	\\
	&&\sum_{\wn=2}^{\numVeh} h_\wn \frac{1}{\spatEigWn} =
	\frac{1}{\weightGain_2}.  \label{eq:sumHiInvL}
\end{IEEEeqnarray}
\label{lem:sumHi}
\end{lemma}
\vspace{-20pt}
\begin{proof}
The terms in
$h_\wn=g_\wn \, \eigVect_{\numVeh\wn }$ can be written as $g_\wn=e_\wn^T
\matEigVect^{-1}e_2$ and $\eigVect_{\numVeh \wn }=e_\numVeh^T \matEigVect$.
Plugging them into (\ref{eq:sumHiL}) yields
\begin{IEEEeqnarray}{rCl}
	\sum_{\wn=2}^{\numVeh} h_\wn \spatEigWn^m &=& e_\numVeh^T {\matEigVect}
	\matJ^m {\matEigVect^{-1}} e_2 
	=
	 e_\numVeh^T \lapl^m e_2 =
	\laplPowEl{\numVeh 2}{m}
\end{IEEEeqnarray}
where $\laplPowEl{ij}{m}$ is the $(i,j)$ element of $\lapl^m$. Laplacian
is a banded matrix with nonzero diagonal and the first subdiagonal and by
powering it, we add new nonzero bands. Hence, $\lapl$ can be powered
at most $\numVeh-3$ times to keep zeros at
$(\numVeh-2)$th subdiagonal and the element
$\laplPowEl{\numVeh 2}{m}=0$ for $m=0,\ldots,N-3$. 
	
Let $\matJ_r =
\matEigVect_r^{-1} \redLapl \matEigVect_r$ be the Jordan form of $\redLapl$.
Equation (\ref{eq:sumHiL}) is obtained in a similar way using
${\matEigVect_r} \matJ_r^{-1} {\matEigVect_r^{-1}}=\redLapl^{-1}$ as
\begin{IEEEeqnarray}{rCl} 
	\sum_{\wn=2}^{\numVeh} h_\wn \frac{1}{\spatEigWn} &=&
	e_{\numVeh-1}^T{\matEigVect_r} \matJ_r^{-1} {\matEigVect_r^{-1}} e_1
	\!=\!\redLaplPowEl{N\!-\!1,1}{-1}
	= \frac{1}{\weightGain_2}. 
\end{IEEEeqnarray}
\end{proof}

\begin{proof}[\textbf{Proof of Theorem \ref{lem:blockProd}}] 
For simplicity, only the case of \textit{non-defective} Laplacian is shown here. First we need to evaluate a characteristic polynomial of $\redLapl$
\begin{equation}
	\det(s\Id_{N-1} + \redLapl) = s^{N-1} + \charPolCoef_{N-2} s^{N-2} + \ldots +
	\charPolCoef_1 s + \charPolCoef_0,
\end{equation}
where $\charPolCoef_{N-2}=\text{Tr}(\redLapl)=\sum_{\wn=2}^{\numVeh}\spatEigWn$ and $\charPolCoef_0 = \prod_{\wn=2}^{\numVeh}\spatEigWn$.
	
The transfer function $\diagTransBlockN(s)$ was defined in
(\ref{eq:transferFunctionNSum}) as
\begin{equation}
	\diagTransBlockN(\lapDom)\!=\!\sum_{\wn=1}^{\numVeh} g_\wn
	\diagTransBlock_i(s) \eigVect_{\numVeh,i}\!=\!\sum_{\wn=1}^{\numVeh} g_\wn
	\eigVect_{\numVeh,i} \frac{\vehNumCoef(\lapDom) \contNumCoef(\lapDom)}{\vehDenCoef(\lapDom) \contDenCoef(\lapDom) + \spatEigWn \vehNumCoef(\lapDom)
	\contNumCoef(\lapDom)}.\label{eq:sumFormTn}
\end{equation}
Since $g_{1}=0$ (the leader
cannot be controlled from the second vehicle), the block corresponding to
$\spatEig_1=0$ does not enter the sum (\ref{eq:sumFormTn}), which then has $N-1$
terms and reads
\begin{IEEEeqnarray}{rCl}
	\diagTransBlockN(s) &=&  \frac{\sum_{\wn=2}^{\numVeh} h_\wn
	\olnum
	\prod_{j=2, j \neq \wn}^{\numVeh} \left[\olden + \spatEig_j \olnum\right] }
	{\prod_{\wn=2}^{\numVeh} \left[ \olden + \spatEigWn \olnum \right]}.
	\label{eq:prodTF}
\end{IEEEeqnarray}
We define $\olden(s)=\vehDenCoef(\lapDom)\contDenCoef(\lapDom)$, 
$\olnum(s)=\vehNumCoef(\lapDom)\contNumCoef(\lapDom)$ and $h_i=g_i
\eigVect_{\numVeh i}$. The argument $s$ is omitted. The numerator of (\ref{eq:prodTF}) is then
	
{
\begin{IEEEeqnarray}{rCl}	
	&&\sum_{i=2}^{\numVeh} h_i \olnum \prod_{\mathclap{j=2, j \neq \wn}}^{\numVeh}
	[\olden + \spatEig_j \olnum] =\sum_{i=2}^{\numVeh} h_i \olnum \Bigg \{ \olden^{N-2}\nonumber 		
	\\ 
	&&
	+
	\left[\olden^{N-3}\olnum \smashoperator{\sum_{j=2, j \neq i}^{\numVeh} }\spatEig_j\right]  
	+ \left[\olden^{N-4}\olnum^{2}\smashoperator{\sum_{j=2, k=2,
	j\neq k \neq i}^{\numVeh}} \spatEig_j\, \spatEig_k \right] \label{eq:numPos}
	 \\
	&&+  \left.\ldots + \left[\olden^{1}\olnum^{N-3}{\sum_{j=2, j \neq
	i}^{\numVeh}}\quad \quad{\prod_{\mathclap{k=2, k \neq i \neq
	j}}^{\numVeh}}\spatEig_k \right] \!+\!\left[\olnum^{N-2}\smashoperator{\prod_{\mathclap{j=2, j \neq
	i}}^{\numVeh}} \spatEig_j\right] \right \}.
	\nonumber
\end{IEEEeqnarray}
\normalsize}
	
Let us put the terms with equal powers of $\olden^{i}\olnum^{j}$ in (\ref{eq:numPos}) together. First, take those with $\olden^{N-2}
\olnum $. The sum $\olden^{N-2}\olnum \sum_{\wn=2}^{N} h_\wn= 0$, using
(\ref{eq:sumHiL}) in Lemma \ref{lem:sumHi}. Second, take those with
$\olden^{N-3}\olnum^2$:
\begin{IEEEeqnarray}{rCl}
	&&\olden^{N-3}\olnum^2 \sum_{\wn=2}^{\numVeh} h_\wn \,\sum_{\mathclap{j=2,
	j\neq \wn}} ^{\numVeh} \spatEig_j 
	= 
	\olden^{N-3}\olnum^2 \sum_{\wn=2}^{\numVeh} h_\wn (\charPolCoef_{N-2}
	- \spatEig_\wn) 	\nonumber	
	\\
	&&
	= \olden^{N-3}\olnum^2 \charPolCoef_{N-2} \sum_{\wn=2}^{\numVeh}
	h_\wn - \olden^{N-3}\olnum^2 \sum_{\wn=2}^{\numVeh}
	h_\wn \spatEigWn = 0.
\end{IEEEeqnarray} 
We used the fact that
$h_\wn (\spatEig_2 + \ldots + \spatEig_{\wn-1} + \spatEig_{\wn+1}
+ \ldots + \spatEig_{\numVeh}) = h_\wn (\charPolCoef_{N-2} - \spatEigWn)$ and then applied Lemma \ref{lem:sumHi}. Using similar constructions we arrive to the fact that all powers of
$\olden^{i}\olnum^{j}$ are multiplied by zero, so they vanish. The only exception is the term
with $\olnum^{N-1}$, for which we get
\begin{equation}
	\olnum^{N-1} \sum_{\wn=2}^{\numVeh}h_\wn \,\, \prod_{\mathclap {j=2, j\neq
	\wn}}^{\numVeh} \spatEigWn = \olnum^{N-1} \charPolCoef_0
	\smashoperator \sum_{\wn=2}^{\numVeh}h_\wn \frac{1}{\spatEigWn} = \olnum^{N-1}
	 \frac{\charPolCoef_0}{\weightGain_2}.
\end{equation}
We used the fact that $\prod_{\wn=2}^{\numVeh}\spatEigWn=\charPolCoef_0$, so
$\prod_{j=2, j\neq \wn}^{\numVeh} \spatEig_j=\frac{\charPolCoef_0}{\spatEig_\wn}$.
The last equality follows from (\ref{eq:sumHiInvL}).	
With all these terms, the fraction in (\ref{eq:prodTF}) simplifies to
\begin{IEEEeqnarray}{rCl}
	\diagTransBlockN(\lapDom) &=&
	\frac{1}{\weightGain_2}\frac{(\vehNumCoef(\lapDom)\contNumCoef(\lapDom))^{N-1}\prod_{\wn=2}^{N}\spatEigWn
	}{\prod_{\wn=2}^{\numVeh}
	\left[ \vehDenCoef(\lapDom)\contDenCoef(\lapDom) + \spatEigWn \vehNumCoef(\lapDom)\contNumCoef(\lapDom) \right]} 
	\nonumber		\\
	&=&
	\frac{1}{\weightGain_2}\prod_{\wn=2}^{\numVeh}
	\frac{\spatEigWn \vehNumCoef(\lapDom)\contNumCoef(\lapDom)}{
	\vehDenCoef(\lapDom)\contDenCoef(\lapDom) + \spatEigWn
	\vehNumCoef(\lapDom)\contNumCoef(\lapDom)}.
\end{IEEEeqnarray}
This is a series interconnection of blocks $\diagTransBlockEigWn(s)$, which
proves the theorem. The cases with defective Laplacian matrices can be treated
in a similar way with all the results valid.
\end{proof}

\ifCLASSOPTIONcaptionsoff
  \newpage
\fi

\bibliographystyle{IEEEtran}  

\bibliography{Papers-TAC-HarmInstab} 

\end{document}

%% file: shared/notation.tex
\newcommand{\pos}{y} 
\newcommand{\inp}{r} 
\newcommand{\totalInput}{u}
\newcommand{\relCoupling}{c}

\newcommand{\dreference}{d_\mathrm{ref}} 

\newcommand{\weightGain}{\mu}

\newcommand{\contNumCoef}{q}
\newcommand{\contDenCoef}{p}
\newcommand{\vehNumCoef}{b}
\newcommand{\vehDenCoef}{a}
\newcommand{\olnum}{\psi} 
\newcommand{\olden}{\phi} 

\newcommand{\lapDom}{s}
\newcommand{\vehicleTf}{G}
\newcommand{\controllerTf}{R}
\newcommand{\systemOutput}{y}
\newcommand{\systemOutputVect}{{\systemOutput}}

\newcommand{\openLoop}{M}
\newcommand{\openLoopGain}{\openLoop_{\min}}

\newcommand{\olOrd}{n}
\newcommand{\wn}{i}

\newcommand{\numVeh}{N}

\newcommand{\lapl}{{L}}
\newcommand{\redLapl}{{\lapl_\mathrm{r}}}
\newcommand{\matPowEl}[3]{(#1^{#3})_{#2}}
\newcommand{\laplPowEl}[2]{\matPowEl{\lapl}{#1}{#2}}
\newcommand{\redLaplPowEl}[2]{\matPowEl{\redLapl}{#1}{#2}}
\newcommand{\spatEig}{\lambda} 
\newcommand{\spatEigWn}{\lambda_\wn}
\newcommand{\spatEigMin}{\spatEig_{\min}}
\newcommand{\spatEigMax}{\spatEig_{\max}}
\newcommand{\eigVect}{v} 

\newcommand{\spatFreqWn}{\theta_\wn} 
\newcommand{\matJ}{{J}}
\newcommand{\matEigVect}{{V}}

\newcommand{\stateVect}{{x}}
\newcommand{\stateVectMod}{{\hat{\stateVect}}}
\newcommand{\inpDiag}{\hat{\inp}} 
\newcommand{\inputCoef}{g}
\newcommand{\inpVect}{{\inp}}
\newcommand{\outVect}{{\pos}}

\newcommand{\stateVectDer}{\dot{\stateVect}}
\newcommand{\stateVectModDer}{\dot{\stateVectMod}}

\newcommand{\matC}{{C}}

\newcommand{\matAi}{{A}}
\newcommand{\matBi}{{B}}
\newcommand{\matCi}{{C}}

\newcommand{\Id}{{I}}
\newcommand{\In}{\Id_\olOrd}
\newcommand{\IN}{\Id_\numVeh}

\newcommand{\diagTransBlock}{{F}}
\newcommand{\diagTransBlockWn}{\diagTransBlock_{\wn}}
\newcommand{\diagTransBlockEig}{T}
\newcommand{\diagTransBlockEigWn}{\diagTransBlockEig_{\spatEigWn}}

\newcommand{\diagTransBlockEigMin}{\diagTransBlockEig_{\spatEig_{\min}}}
\newcommand{\diagTransBlockN}{\diagTransBlockEig_{\numVeh}}

\newcommand{\charPolCoef}{\alpha}



%% file: shared/MyTikzStyles.tex
\tikzstyle{genGraphNode} = [minimum size=3mm, thick, node
distance=5mm] 
\tikzset{
    >=stealth', bend angle=10 }
\tikzstyle{normalNode} = [genGraphNode, circle,draw=black,fill=black!10,thick]
\tikzstyle{controllingNode} = [genGraphNode,
circle,draw=black!20,fill=black!60,very thick, text=white, font=\bfseries] 
\tikzstyle{observingNode} =[genGraphNode, circle,draw=black!50,fill=white,thick]
\tikzstyle{placeholder}=[genGraphNode, circle]

\tikzstyle{block} = [draw, fill=white!20, rectangle, 
    minimum height=2em, minimum width=3em] 
\tikzstyle{gain} = [draw,regular polygon,
        regular polygon sides=3, minimum height=2em, minimum
width=1em, inner sep=0pt] 
\tikzstyle{gainLeft} = [gain, shape border rotate=-90]
\tikzstyle{gainRight} = [gain, shape border rotate=90]
\tikzstyle{gainDown} = [gain, shape border rotate=180]
\tikzstyle{gainUp} = [gain, shape border rotate=90]
\tikzstyle{sum} = [draw, fill=white!20, circle,inner sep=1pt]
\tikzstyle{input} = [coordinate]
\tikzstyle{output} = [coordinate]
\tikzstyle{pinstyle} = [pin edge={to-,thin,black}]
\tikzstyle{junction} = [draw, circle, minimum height=0.02em, fill=black, inner sep=0pt]

%% file: figures/ClosedLoopSeriesSmall.tex
\small
\begin{tikzpicture}[auto,>=latex',node distance=0.6cm and 0.4cm]
		\node [input] (input) {};
		\node [junction, below=of input, below=1em, label=left:$\inp_{2}$] (juncInp) {};
		\node [gainDown, below=of juncInp, below=1.5em, label=right:$\inputCoef_{i+1}$] (inpGain1) {};
		\node [sum, below=of inpGain1, below=1em] (sum) {{+}}; 
	    \node [block, right=of sum] (controller) {$\controllerTf(s)$};
	    \node [block, right=of controller] (system) {$\vehicleTf(s)$};
	    \node [junction, right=of system,label=above:$\hat{\pos}_{i+1}$ ] (y1) {};
	    \node [junction, below=of y1, below=2.5em] (juncOut1) {}; 
	    \draw [->] (controller) -- node[name=u] {} (system);
	    \node [gainRight, left=of juncOut1, left=4em] (gain) {$\spatEig_i$};
	    
	    \draw [->] (gain) -| node[pos=0.94] {$-$} 
        node [near end] {} (sum);  
        \draw [->] (sum) to (controller);
        \draw [-] (input) to (juncInp);
        \draw [->] (juncInp) to (inpGain1);
        \draw [->] (inpGain1) to node[name=ud1, swap] {$\inpDiag_{i+1}$}(sum);
        
    	\node [sum, right=of system, right=0.8cm] (sum2) {{+}};
    	\node [gainDown, above=of sum2, above=1em] (inpGain2) {$\inputCoef_{i}$};
    	\node [block, right=of sum2] (controller2) {$\controllerTf(s)$};
    	\node [block, right=of controller2] (system2) {$\vehicleTf(s)$};
    	\draw [->] (controller2) to node[name=u2] {} (system2);

    	\node [junction, right=of system2, label=above:$\hat{\pos}_{i}$] (y2) {};
    	\node [output, right=of y2, right=0.6cm] (output)  {}; 
    	\draw [-] (system2) to (y2);
    	\draw [->] (y2) to (output);
    	\node [junction, below=of y2, below=2.5em] (juncOut2) {}; 
    	\node [gainRight, left=of juncOut2, left=4em] (gain2) {$\spatEig_i$};
    	\draw [->] (juncOut2) -- (gain2); 
    	\draw [->] (gain2) -| node[pos=0.94] {$-$} 
        node [near end] {} (sum2); 
        \draw [->] (sum2) to (controller2);
        \draw [->] (inpGain2) -- node[name=ud2]{$\inpDiag_i$} (sum2);
        \draw [->] (juncInp) -| (inpGain2);
        \draw [-] (system) to (y1);
        \draw [->] (y1) to (sum2);
        \draw [->] (y1) |- (gain);
        
       \node [gainDown, below=of juncOut1, below=0.7em, label=right:$\eigVect_{\numVeh, i+1}$] (gainOut1) {};
       \node [gainDown, below=of juncOut2, below=0.7em, label=left:$\eigVect_{\numVeh, i}$] (gainOut2) {};
       \node [sum, below= of gainOut1, below=0.5em] (sum3) {+};
       \node [output, left=of sum3, left=3em] (outputFin) {$\pos_{\numVeh}$};
       \draw [->] (y2) to (gainOut2);
       \draw [->] (y1) to (gainOut1);
       \draw [->] (gainOut1) to (sum3);
       \draw [->] (gainOut2) |- (sum3);
       \draw [->] (sum3) to node[name=yFin, swap]{$\pos_{\numVeh}$}(outputFin);
\end{tikzpicture}
\normalsize